\documentclass[a4paper,12pt]{amsart}
\usepackage{graphicx}
\usepackage{color}

\newcommand{\R}{\mathbb R}
\newcommand{\C}{\mathbb C}

\newcommand{\beq}{\begin{equation}
}
\newcommand{\eeq}{\end{equation}}
\newcommand{\pr}{^{\prime}}

\newtheorem{thm}{Theorem}[section]
\newtheorem{cor}[thm]{Corollary}
\newtheorem{lem}[thm]{Lemma}
\newtheorem{prop}[thm]{Proposition}
\newtheorem{defn}[thm]{Definition}
\theoremstyle{remark}
\newtheorem{rem}[thm]{Remark}
\numberwithin{equation}{section} \makeatletter
\def\@cite#1#2{#1\if@tempswa , #2\fi}
\def\@biblabel#1{$^{\hbox{\scriptsize{#1}}}$}
\makeatother

\begin{document}

\title[]{Edge of the Wedge Theorem for Tempered Ultrahyperfunctions}

\author{E. Br\"uning \and S. Nagamachi }

\address[E. Br\"uning]{School of Mathematical Sciences,
University of KwaZulu-Natal, Private Bag X54001, Durban 4000,
South Africa} \email{bruninge@ukzn.ac.za}

\address[S. Nagamachi]{
Emeritus professor, Faculty of
Engineering, The University of Tokushima\\ Tokushima 770-8506,
Japan} \email{shigeaki-matematiko@memoad.jp}

\begin{abstract}
Tempered ultrahyperfunctions do not have the same type of localization properties as Schwartz distributions or Sato hyperfunctions; but the localization properties seem to play an important role in the proofs of the various versions of the edge of the wedge theorem. Thus, for tempered ultrahyperfunctions, one finds a global form of this result in the literature, but no local version.

 In this paper we propose and prove a formulation of the edge of the wedge theorem for tempered ultrahyperfunctions, both in global and local form. We explain our strategy first for the one variable case. We argue that in view of the cohomological definition of hyperfunctions and ultrahyperfunctions, the global form of the edge of the wedge theorem is not surprising at all.
\end{abstract}

\maketitle \tableofcontents

\section{Introduction}
Nowadays, i.e., about fifty years after its discovery, there are
many versions of the `edge of the wedge theorem' which originated
through the challenges of relativistic quantum field theory and
the theory of dispersion relations for scattering amplitudes (see
[\cite{BV60}]). In quantum field theory the important fact that
the Wightman functions are holomorphic in the region of all
totally space-like points is
 shown  by a simple application of the edge of the wedge theorem.

The statements in these various versions of the edge of the wedge
theorem assert the extendability of holomorphic functions defined
in wedges in complex space $\C^n$ with edge in real space $\R^n$,
under certain conditions.

Recall the original version of Bogoliubov [\cite{BV60}]:
\begin{thm}[Bogoliubov] Let $C \subset \R^n$ be an open {proper} cone with vertex at the origin and denote by $C_r = C \cap B(0,r)$ the intersection of $C$ with the open ball of radius $r>0$ centered at the origin; for an open nonempty set $E \subset \R^n$ introduce the wedges $W^{\pm}=E \pm i C_r$ with common edge $E$. If now $F_1$ is a holomorphic function on $W^+$ and $F_2$ is holomorphic on $W^-$ and if $F_1$ and $F_2$ have the same boundary values on $E$,
\beq \label{equal-bv} \lim_{\stackrel{y \to 0}{y \in C_r}} F_1(x
+i y)= \lim_{\stackrel{y \to 0}{y \in C_r}}F_2(x -i y), \qquad x
\in E, \eeq then $F_1$ and $F_2$ can be extended holomorphically
to a complex neighborhood $\Omega$ of $W^+ \cup E \cup W^-$.
\end{thm}
Naturally, in (\ref{equal-bv}) it is important to specify in which
sense the boundary values are considered. In  Bogoliubov's version
these boundary values are taken in the sense of Schwartz
distributions. Note that for $n=1$ and when the boundary values
are taken in the sense of continuous functions, this result is
easily proven by using Morera's theorem.

Over the last fifty years, this result has been extended in
several directions:
\begin{itemize}
 \item the type of generalized functions for the boundary values in (\ref{equal-bv}), e. g., Schwartz distributions, Sato hyperfunctions, Fourier hyperfunctions, and ultradistributions;
 \item the number $m$ of wedges, $m>2$;
 \item the `topological nature' of the {edge} $E$, e. g., an open nonempty set or a maximal real submanifold.
\end{itemize}

We comment here on the case of Fourier hyperfunctions. A Fourier
hyperfunction $f$ has two realizations.  One is as a dual element
of the test-function space ${\mathcal O}
_{_{_{\hspace{-2.5mm}\sim}}}{} (\boldsymbol{D}^{n})$ and the other
is as a formal sum
$$      f(x) = \sum _{j=1}^{m} F_{j}(z)$$
where $F_{j}(z)$ is holomorphic in a wedge $W_{j} = \boldsymbol{D}^{n} + i\Gamma _{j}$, that is, an element of the relative {\v Cech }
cohomology group {$H^{n}(\mathfrak{W}, \mathfrak{W} \pr ; \tilde{{\mathcal O} } )$ which is isomorphic to} $H^{n}_{\boldsymbol{D}^{n}}(\boldsymbol{Q}^{n};
\tilde{{\mathcal O} } )$ {for a suitable relative covering $(\mathfrak{W}, \mathfrak{W} \pr )$ }of the pair $(\boldsymbol{Q}^{n},
\boldsymbol{Q}^{n}\backslash \boldsymbol{D}^{n})$, {where $\tilde{{\mathcal O} } $ is the sheaf} of slowly increasing holomorphic
functions on $\boldsymbol{Q}^{n} = \boldsymbol{D}^{n} + i
\mathbb{R}^n$ and $\boldsymbol{D}^{n} = \mathbb{R}^n \cup
S^{n-1}_{\infty }$ is the radial compactification of
$\mathbb{R}^n$ (see [\cite{Ka70, BN89}]). In the case of
hyperfunctions and Fourier hyperfunctions, the edge of the wedge
theorem tells us when the above sum is zero. Note that {$H^{n}(\mathfrak{W}, \mathfrak{W} \pr ; \tilde{{\mathcal O} } )$}
can be expressed as the following quotient space (see [\cite{Ka88}])
\begin{equation} \label{quotient}
   \oplus _{j=1}^{m} \tilde{{\mathcal O} } (W_{j})/[\oplus _{j<k}
\tilde{{\mathcal O} } (W_{j}+W_{k})].
\end{equation}
 This denominator appears in the general edge of the wedge theorem for Fourier hyperfunctions (see [\cite{NN93}]).
 \begin{rem}
Relation (\ref{quotient}) `contains' the most important versions
of the EOW: If $\displaystyle \sum _{j = 1}^{m}F_{j}(z) = f(x) =
0$, then there exist functions $H_{jk} \in \tilde{{\mathcal O} }
(W_{j} + W_{k})$ for $j < k$ such that
\begin{equation} \label{formalsum}
      F_{j}(z) = \sum _{k=1}^{m}H_{jk}(z), \  j = 1, \ldots ,
m,
\end{equation}
where we put
\begin{equation} \label{antisymmetric}
      H_{jk}(z) = -H_{kj}(z) \  {\rm for \, }\  j > k \  {\rm and
\, }\  H_{jj}(z) = 0.
\end{equation}
The above statement is Martineau's version of the EOW theorem
[\cite{Ma64, Ma67, Ma70}].

Note that conversely, for the functions $F_{j}$  defined by
(\ref{formalsum}) with functions $H_{jk}$ satisfying relation
(\ref{antisymmetric}), the sum
$$      \sum _{j = 1}^{m}F_{j}(z) = \sum _{j = 1}^{m} \sum _{k=1}^{m}H_{jk}(z)$$
is reduced to zero  and defines the zero Fourier hyperfunction.

When $m = 2$, we have Epstein's version of the EOW theorem
[\cite{Ep60}], i.e., if $F_{j} \in \tilde{{\mathcal O} } (W_{j})$
$(j = 1, 2)$ define the same Fourier hyperfunction, then $F_{j}
\in  \tilde{{\mathcal O} } (W_{j})$ $(j = 1, 2)$ are analytically
continued to $H_{12} \in \tilde{{\mathcal O} } (W_{1} + W_{2})$.
(Epstein's original version is formulated in terms of Schwartz
distributions as boundary values).

 If $\Gamma _{2} = -\Gamma _{1}$ and if the boundary values are taken in the sense of Schwartz distributions, then we have Bogoliubov's version of the EOW theorem.
 \end{rem}
In our recent investigations of relativistic quantum field theory
with a fundamental length (see [\cite{BN04,BN08,NB08a,NB09}]) we
need a version of the edge of the wedge theorem for tempered
ultra-hyperfunctions and it is this version which is treated in
this article.

\section{Preliminaries}
\subsection{Global and local versions of the EOW Theorem }
In these preliminary considerations we put the global and the
local forms of the edge of the wedge theorems for hyperfunctions
and ultra-hyperfunctions into the perspective of cohomology
theory.

Let us recall the global form of this theorem for hyperfunctions
(For simplicity and our intended application to quantum field
theory, we consider in this paper only Bogoliubov's version for
$\Gamma = V_+$):

 Let $V_{+}\subset \R^4$ denote the forward light-cone; suppose
that $F_{1}$ is an analytic function in $T(V_{+}) = \mathbb R^{4}
+ iV_{+}$ and $F_{2}$ an analytic function in $T(-V_{+})$.  Then
the two functions $F_{i}$ $(i = 1, 2)$ define hyperfunctions
$f_{i}$ on $\R^4$.  If $f_{1} = f_{2}$ , then $F_{i}$ are
analytically continued to an entire function $F$.

The local form of the EOW theorem for hyperfunction can be
formulated as follows. Let $U$ be an open set in $\mathbb R^{n}$
and $V$ an open set in $\mathbb C^{n}$ such that $U = V \cap
\mathbb R^{n}$.  Then we have the canonical restriction map
\begin{equation} \label{restriction}
     H^{n}_{ \mathbb R^{n}}(\mathbb C^{n}, {\mathcal O} )
\rightarrow  H^{n}_{U}(V, {\mathcal O} ) = H^{n-1}(V \backslash U,
{\mathcal O} ),
\end{equation}
where $\mathcal{O}$ is the sheaf of holomorphic functions on
$\mathbb{C}^n$.  $H^{n}_{U}(V, {\mathcal O} )$ is independent of
the complex neighborhood $V$ of $U$ by the excision theorem, and
the presheaf $\{ U \rightarrow H^{n}_{U}(V, {\mathcal O} )\} $ is
the sheaf of hyperfunctions on $\mathbb{R}^n$ which is often
denoted by $\mathcal{B}$. The local form of the EOW theorem now
reads (we use the notation from above):

If $f_{1} =f_{2}$ in $U$ (or the restrictions $f_{j\vert U}$ $(j =
1, 2)$ coincide or the support of $f_{1} - f_{2}$ is contained in
the complement of $U$), then $F_{j}$ are analytically continued to
each other through $U$.

There are two ways to treat the EOW theorem; the functional method
(see [\cite{Hoe83}]) and the cohomological method (see
[\cite{Mo73, Ka88}]).  The functional method uses the notion of
the analytic wave front set $WF_{a}(f)$ of hyperfunctions $f$ on
$\mathbb R^{n} \times  (\mathbb R^{n}\backslash \{ 0\} )$ and the
decomposition of $WF_{a}(f)$.  The colomological method uses the
notion of the flabby sheaf ${\mathcal C} $ of micro functions on
$\mathbb R^{n} \times  S^{n-1}$ and the exact sequence
$$      0 \rightarrow  {\mathcal A}  \rightarrow  {\mathcal B}  \rightarrow  \pi _{*} {\mathcal C}  \rightarrow  0,$$
where ${\mathcal A} $ is the sheaf of real analytic functions on
$\mathbb R^{n}$ and $\pi _{*} {\mathcal C} $ is the direct image
of ${\mathcal C} $ under the projection $\pi : \mathbb R^{n}
\times S^{n-1} \rightarrow \mathbb R^{n}$, i.e., the sheaf on
$\mathbb R^{n}$ defined by the correspondence
$$      \mathbb R^{n} \supset  U \rightarrow  {\mathcal C} (\pi ^{-1}(U)).$$
Finally we comment on the  difficulties for the EOW for tempered
ultra-hyperfunctions.

{Any element of ${\mathcal T} (T(\mathbb {R}^{n}))^{\prime } $ belongs to some ${\mathcal T} (T(K))^{\prime }$, $K = [-k, k]^{n}$ for suitable $k > 0$ and
\begin{equation} \label{t-uhf}
{\mathcal T} (T(K))^{\prime } = H^{n}(\mathfrak{W}, \mathfrak{W} \pr ; {\mathcal O}_0 ),
\end{equation}
where $\mathfrak{W} = \mathfrak{W} \pr \cup \{\mathbb {C}^{n} \}$ and $\mathfrak{W} \pr = \{T(E_j); j = 1, \ldots , n \}$, $E_j = \{y \in \mathbb R^{n}; \vert y_j \vert > k \}$
 are relative covering of $(\mathbb {C}^{n}, \mathbb {C}^{n} \setminus T(K)$}
(see [\cite{Ha61, Mo75b}]).  Hyperfunctions are localized in a relatively
open set $U$ of the closed set $\mathbb {R}^{n}$ in $\mathbb {C}^{n}$
by the formula (\ref{restriction}).  On the  other hand,  tempered
ultra-hyperfunctions may be localized in a relatively open set $U$
of the closed set $T(K)$ of $\mathbb C^{n}$, but not in a
relatively open set $U$ of the closed set $\mathbb R^{n}$ in
$\mathbb C^{n}$.  This is the reason why  tempered
ultra-hyperfunctions have no (standard) localization property in
$\mathbb{R}^n$. In [\cite{BN04, NB08a, BN08}] this property
 has been successfully applied to axiomatic
quantum field theory in order to formulate such a theory with a
fundamental length.

The global form of the edge of the wedge theorem for tempered
ultra-hyperfunctions reads: {Let $\Gamma = \{y \in \mathbb{R}^n ; y_j > k, j = 1, \ldots , n\}$.}
If $F_{1}$ is a polynomially increasing holomorphic function in
$T(\Gamma )$ and $F_{2}$ a polynomially increasing holomorphic
function in $T(-\Gamma )$, then these functions $F_{i}$ $(i = 1,
2)$ define tempered ultra-hyperfunctions $f_{i}$ and if $f_{1} =
f_{2}$, then $F_{i}$ are analytically continued to a
{polynomial $F$ (see [\cite{Mo75b}])}.

\begin{rem}
It is interesting to note that even though $T(\Gamma )$ and
$T(-\Gamma )$ are separated by a gap of size $2\sqrt{n} k$ the
functions $F_{1}$ and $F_{2}$ are analytically continued to each
other!
\end{rem}
At first sight this seems to be quite surprising. However, from
 the  point of view of the cohomological definition of  ultra-hyperfunctions, this result is not so surprising, since
{$H^{n}(\mathfrak{W}, \mathfrak{W} \pr ; {\mathcal O}_0 )$ has the following representation
$$
\mathcal{O}_0 (\mathbb{C}^n \# T(K))/ \sum _{j=1}^{n} \mathcal{O}_0 (W_{\hat{j}}),
$$
where
$$
   \mathbb{C}^n \# T(K)) = T(E_1) \cap \cdots \cap T(E_n),
$$
$$
    W_{\hat{j}} = T(E_1) \cap \cdots \cap \widehat{T(E_j)} \cap \cdots \cap T(E_n)
$$}and the
denominator of this representation then shows this result. We
explain this in more detail in the next subsection for the one
dimensional case. The cohomological treatment of
ultra-hyperfunctions is given in [\cite{Mo69, Mo72}] and the above
representation was presented in {[\cite{Mo75b}]}.  The global form of
the EOW theorem for tempered ultra-hyperfunctions has  been
shown in [\cite{Su04}] and some preliminary version in
[\cite{Fr06, Fr07}].

In the functional method [\cite{Hoe83}], H\"ormander used a kernel
$K(z)$ defined by \beq \label{hkernel}
     K(z) = (2\pi )^{-n}
\int e^{i\langle z, \xi \rangle }/I(\xi ) d\xi , \ I(\xi ) = \int
_{\vert \omega \vert =1} e^{-\langle \omega , \xi \rangle }
d\omega \eeq to prove the edge of the wedge theorem.  For the
proof of the local form of the EOW theorem for tempered
ultra-hyperfunctions we also use the functional method with some
modification $K_{r}(z) = r^{-n}K(z/r)$ of this kernel $K(z)$ for
$r
> 0$.

 The local form of the edge of the wedge theorem for hyperfunction has the following formulation ($F_j$ and $f_j$ are related as in the above results):
 If $f_{1} =f_{2}$ in an open set $O$, then $F_{1}$ and $F_{2}$ are
analytically continued to each other through $O$.

 Since tempered ultra-hyperfunctions have no localization property, it is not easy to formulate a local form of the edge of the wedge theorem.
In this paper, we suggest a formulation of the local version of
the edge of the wedge theorem by using the notion of a carrier.
\begin{rem}
For the one dimensional case, the Cauchy-Hilbert transformation
(\ref{chtransf}) gives the isomorphism of the space ${\mathcal O}
(L)^{\prime }$ of analytic functionals with carriers in $L$ onto
the relative cohomology group of covering (\ref{cohomology-u}) of
the pair $(\mathbb C, \mathbb C\backslash L)$.  But for the multi
dimensional case, we need the additional assumption to have the
expression of the space ${\mathcal O} (L)^{\prime }$ of analytic
functionals with carriers in a compact set $L \subset  \mathbb
C^{n}$ as a collection of holomorphic functions.  In fact, the
isomorphism of the space ${\mathcal O} (L)^{\prime }$ of analytic
functionals with carriers in a compact set $L \subset \mathbb{
C}^{n}$ onto the relative cohomology group $H^{n}_{L}(\mathbb
{C}^{n}, {\mathcal O} )$ is proven under the condition a) $H^{p}(L,
{\mathcal O} ) = 0$ for $p = 1, 2, \ldots $ (see [\cite{Mo75b}]), and the isomorphism
of $H^{n}_{L}(\mathbb {C}^{n}, {\mathcal O} )$ onto the cohomology
group $H^{n}(\mathfrak{W}, \mathfrak{W}^{\prime }; {\mathcal O} )$ of a
covering $(\mathfrak{W}, \mathfrak{W}^{\prime })$ of the pair $(\mathbb {C}^{n}, \mathbb
{C}^{n}\setminus L)$ is proven under the condition b)
$H^{p}(W_{\lambda _{0}}\cap  \ldots \cap W_{\lambda _{m}}) = 0$
for $p \geq  1$ and $W_{\lambda _{0}}, \ldots , W_{\lambda _{m}}
\in  \mathfrak{W}$ {(Leray theorem, see [\cite{Mo93}])}.   In the case of hyperfunctions, a compact set $L \subset
\mathbb {R}^{n}$ satisfies  condition a) and there exists a relative
covering $(\mathfrak{W}, \mathfrak{W}^{\prime })$ of the pair $(\mathbb {C}^{n}, \mathbb
{C}^{n}\setminus L)$ which satisfies
 condition b) (see [\cite{Ka88, Mo93}]).  But in the case of
analytic functional ${\mathcal O} (L)^{\prime }$, conditions a)
and b) may not necessarily be satisfied for some compact set $L
\subset \mathbb {C}^{n}$ and any relative covering $(\mathfrak{W}, \mathfrak{W}^{\prime
})$ of the pair $(\mathbb {C}^{n}, \mathbb {C}^{n}\backslash L)$.
Therefore, in this paper we employ the functional method.
\end{rem}

The main result in this regard is Corollary \ref{application}
which has an intimate connection to axiomatic quantum field theory
with a fundamental length (see [\cite{NB09}]).  This corollary
says that if $f_{1} - f_{2} \in {\mathcal T} (L)^{\prime }$ for
some $\ell $-neighborhood
$$      L = \{ w \in  \mathbb C^{4}; \exists\, x \in  V \  \vert {\rm Re \, }w - x\vert  + \vert {\rm Im \, }w\vert _{1} < \ell \} ,$$
of the light-cone $V$, $\ell >0$, then $F_{i}$ $(i = 1, 2)$ are
analytically continued to each other through a set
$$      \{ x \in  \mathbb R^{4}; {\rm dist \, }(x, V) > (\sqrt{2} + 1)\ell \} .$$

\subsection{The one dimensional case}
In order to explain the basic idea of our strategy of proof for
the local version of the EOW theorem for tempered
ultrahyperfunctions, we illustrate it here for the technically
much simpler case of one dimension. And we prepare this with
explaining the proof for hyperfunctions in one variable.

The space of hyperfunctions of one variable is the quotient space
$$      {\mathcal B} (\mathbb R) = {\mathcal O} (\mathbb C\backslash  \mathbb R)/{\mathcal O} (\mathbb C).$$
Let $F_{1}$ (resp. $F_{2}$) be a holomorphic function in the upper
(resp. lower) half plane.  Then the  pair of functions $(F_{1},
F_{2}) \in  {\mathcal O} (\mathbb C\backslash \mathbb R)$ defines
an element $f (= F_{1} - F_{2})$ of ${\mathcal B} (\mathbb R)$. If
$f = 0$ ($F_{1} = F_{2}$), then $(F_{1}, F_{2}) \in  {\mathcal O}
(\mathbb C)$.  This shows that $F_{1},F_2$  coincide with an
element $F$ of ${\mathcal O} (\mathbb C)$.  Thus the EOW theorem
automatically follows from the cohomological definition of
hyperfunctions.

The local version of the EOW theorem for hyperfunctions is also a
direct consequence of the cohomological definition of
hyperfunctions on an open set $U$ of $\mathbb R$:
$$ {\mathcal B}(U) = {\mathcal O} (V\backslash U)/{\mathcal O} (V),$$
where $V$ is a complex neighborhood of $U$ such that $U = V \cap
\mathbb R$.

Now consider the case of tempered ultra-hyperfunctions. Denote $K
= [-k, k]$ for $k > 0$ and $T(K) = \{ z \in  \mathbb C; \vert {\rm
Im \, }z\vert  \leq  k\} $. Note that  formula (\ref{t-uhf})
shows that any tempered ultra-hyperfunction $f \in {\mathcal T}
(T(\mathbb R))^{\prime }$ can be expressed as an element of the
space
$$      H^{1}_{T(K)}(\mathbb C, {\mathcal O} _{0}) \cong {\mathcal O} _{0}(\mathbb C\backslash T(K))/{\mathcal O} _{0}(\mathbb C)$$
for some $k > 0$.  Let $F_{1}$ (resp. $F_{2}$) be a polynomially
increasing holomorphic function in $\{ z \in \mathbb C; {\rm Im \,
}z
> k\} $ (resp. $\{ z \in \mathbb C; -{\rm Im \, }z > k\} $).  Then
the  pair of functions $(F_{1}, F_{2}) \in  {\mathcal O}
_{0}(\mathbb C\backslash T(K))$ defines an element $f$ of
$H^{1}_{T(K)}(\mathbb C, {\mathcal O} _{0})$. If $f = 0$, then
$(F_{1}, F_{2}) \in {\mathcal O} _{0}(\mathbb C)$.  This shows
that $F_{1},F_2$ coincide with an element $F$ of ${\mathcal O}
_{0}(\mathbb C)$. Thus the global form of the EOW theorem for
tempered ultra-hyperfunctions automatically follows from the
cohomological definition of tempered ultra-hyperfunctions.

Since the notion of localization for ultra-hyperfunctions is not
available in the above sense, there is no literature about the
local version of the EOW theorem for ultra-hyperfunctions. The
notion of localization for generalized functions has an intimate
connection with the notion of support. However,
ultra-hyperfunctions have no supports in general, but they are a
special kind of analytic functionals and have carriers.
\begin{defn}  Let $L$ be a compact set in $\mathbb{C}$.  $L$ is
called a {\bf carrier of an analytic functional} $f$ (a continuous
linear functional on the space $\mathcal{O}(\mathbb{C})$ of entire
functions), if $f$ satisfies
$$  \vert f(\phi ) \vert \leq C_V \sup _{z \in V} \vert \phi (z)
\vert $$ for any open neighborhood $V$ of $L$.
\end{defn}

 Let $K = [a, b] \subset  \mathbb R$ and $H^{1}_{K}(\mathbb C,
{\mathcal O} )$ be the space of hyperfunctions with supports in
$K$, which is isomorphic to the space of analytic functionals with
carriers in $K$:
$$      H^{1}_{K}(\mathbb C, {\mathcal O} ) \cong {\mathcal O} (\mathbb C\backslash K)/{\mathcal O} (\mathbb C).$$
Every $F \in  {\mathcal O} (\mathbb C\backslash K)$ defines a
functional on the space of functions $\phi$ which are holomorphic
in a complex neighborhood $V$ of $K$ by the formula
\begin{equation} \label{functional}
  \phi  \rightarrow  f(\phi ) = -\int _{C} F(z)\phi (z) dz,
\end{equation}
where $C$ is a closed path that encircles $K$ once in the positive
direction.

Let $E_{\xi }^{t}(z) = (4\pi t)^{-1/2}e^{-(z - \xi )^{2}/4t}$ for
$\xi  \in  \mathbb R\backslash K$. Clearly $E_{\xi }^{t}$ is
holomorphic in a complex neighborhood $V$ of $K$ and, as $t
\rightarrow 0$, $E_{\xi }^{t} \rightarrow  \delta (x - \xi )$ and
one has the estimate
\begin{equation} \label{delta}
 \vert f(E_{\xi }^{t})\vert  \leq  M_{C} \sup _{x+iy \in C} (4\pi
t)^{-1/2} e^{(y^{2}-(x-\xi )^{2})/4t} \rightarrow  0
\end{equation}
as $t \rightarrow  0+$, since $C$ can be chosen arbitrary close to
$K$.

Let $L = [a, b]+i[-\ell , \ell ]$ and $f$ be an analytic
functional with carrier $L$.  Then the Cauchy-Hilbert
transformation $F$ of $f$ is defined by \beq \label{chtransf}
 F(z) = (2\pi i)^{-1} f_{w}(1/(w - z));
\eeq it is an element of \beq \label{cohomology-u}
   H^{1}_{L}(\mathbb C, {\mathcal O}_0 )
\cong {\mathcal O}_0 (\mathbb C\backslash L)/{\mathcal O}_0
(\mathbb C), \eeq i.e., the space of tempered ultra-hyperfunctions
with carriers in $L$, and it reproduces the functional $f$ by
formula (\ref{functional}) with a closed path $C$ that encircles
$L$, i.e.,
$$     -\int _{C} F(z)\phi (z) dz =  f_{w} \left( (2\pi i)^{-1} \int _{C}  \phi (z)/(z - w) dz \right) = f(\phi ). $$
For a tempered ultra-hyperfunction $f$ with carrier $L$, we have
(\ref{delta}) if $\xi \in \mathbb R\backslash [a-\ell , b+\ell ]$.
This means that $\xi $ is considered to be outside of $[a, b]$
only if $\xi  \in \mathbb R\backslash [a-\ell , b+\ell ]$!
\begin{rem} Heuristically we read this fact as follows:\\
An ultra-hyperfunction $f$ becomes `aware' of $\xi $ being outside
of $[a, b]$ only if $\xi \in \mathbb R\backslash [a-\ell , b+\ell
]$, while a hyperfunction can be `aware' of $\xi $ being outside
of $[a, b]$ if $\xi \in \mathbb R\backslash [a , b]$.

This difference can also be understood  through cohomological
considerations.  Let $U$ be an open set in $\mathbb{R}$ such that
$[a, b] \cap U = \emptyset $ and $U = V \cap \mathbb{R}$ for an
open set $V$ in $\mathbb{C}$. In the case of hyperfunctions, we
have the restriction
$$     H^{1}_{K}(\mathbb C, {\mathcal O} ) \rightarrow  H^{1}_{U}(V, {\mathcal O} ) \cong {\mathcal O} (V \backslash K)/{\mathcal
O} (V) = {\mathcal O} (V)/{\mathcal O} (V) = 0.$$  But for the
case of ultra-hyperfunctions, we have
$$     H^{1}_{L}(\mathbb C, {\mathcal O}_0 ) \rightarrow H^{1}_{U}(V, {\mathcal O}_0 ) \cong {\mathcal O}_0 (V \backslash L)/{\mathcal O}_0 (V),$$
and there exists an open set $V$ in $\mathbb{C}$ satisfying $[a,
b] \cap U = \emptyset $ and $U = V \cap \mathbb{R}$ but $V
\backslash L
 \neq V$ and ${\mathcal O}_0 (V \backslash L)/{\mathcal O}_0 (V) \neq
 0$.
\end{rem}

Now we study the EOW theorem by using  H\"ormander's kernel
$K(z)$. First, recall that Dirac's $\delta $ function can be
expressed as follows:
$$      \delta (x) = (2\pi )^{-1} \int ^{\infty }_{-\infty }e^{i\xi x} d\xi $$
$$      = (2\pi )^{-1} \lim _{\epsilon \rightarrow 0} \int ^{\infty }_{-\infty }e^{i\xi (x+i\epsilon )}\frac{e^{\xi }}{e^{\xi }+e^{-\xi }} d\xi + (2\pi )^{-1} \lim _{\epsilon \rightarrow 0} \int ^{\infty
}_{-\infty }e^{i\xi (x-i\epsilon )}\frac{e^{-\xi }}{e^{\xi
}+e^{-\xi }} d\xi $$
$$      = (2\pi )^{-1} \lim _{\epsilon \rightarrow 0} \int ^{\infty }_{-\infty }\frac{e^{i\xi (x+i\epsilon -i)}}{e^{\xi }+e^{-\xi }} d\xi + (2\pi )^{-1} \lim _{\epsilon \rightarrow 0} \int ^{\infty
}_{-\infty }\frac{e^{i\xi (x-i\epsilon +i)}}{e^{\xi }+e^{-\xi }}
d\xi $$
$$      = (2\pi )^{-1} \lim _{\epsilon \rightarrow 0} \sum _{\omega =\pm 1} \int ^{\infty }_{-\infty }\frac{e^{i\xi (x+i(1-\epsilon )\omega )}}{e^{\xi }+e^{-\xi }}d\xi .$$
Now introduce the function $K(z)$ of (\ref{hkernel}) for $n = 1$:
$$      K(z) = (2\pi )^{-1} \int ^{\infty }_{-\infty } e^{i\xi z} \frac{1}{2}{\rm {}sech \, }\xi  d\xi
     = \frac{1}{4} {\rm {}sech \, }(\pi z/2)$$
for $\displaystyle {\rm {}sech \, }\xi  = \frac{2}{e^{\xi
}+e^{-\xi }}$.  Then the above representation of the Dirac's delta
function can be rewritten as
$$      \delta (x) = \lim _{\epsilon \rightarrow +0}[K(x+i\epsilon -i)+K(x-i\epsilon +i)]$$
$$      = \lim _{\epsilon \rightarrow +0}(1/4)[{\rm {}sech \, }\pi (x+i\epsilon -i)/2 + {\rm {}sech \, }\pi (x-i\epsilon +i)/2]$$
\beq \label{d-delta}
    = \lim _{\epsilon \rightarrow +0}(1/4)i[{\rm {}cosech \, }\pi
(x+i\epsilon )/2 - {\rm {}cosech \, }\pi (x-i\epsilon )/2], \eeq
where $\displaystyle {\rm {}cosech \, }\xi  = \frac{2}{e^{\xi
}-e^{-\xi }}$.  Since the difference between $(1/4)i{\rm cosech \,
}\pi z/2$ and $-1/(2\pi i z)$ is holomorphic function in $\{z \in
\mathbb{C}; \vert {\rm Im \, } z \vert < 1\}$, formula
(\ref{d-delta}) is equivalent to the famous formula
$$     \delta (x) = -(2\pi i)^{-1}\lim_{\epsilon \rightarrow +0} [1/(x+i\epsilon ) - 1/(x-i\epsilon )].$$
The singular points of ${\rm {}sech \, }\xi $ are $\xi  = i(1 +
2n)\pi /2$ $n = 0, \pm 1, \pm 2, \ldots $ and those of $K(z)$ are
$z = i(1 + 2n)$ $n = 0, \pm 1, \pm 2, \ldots $. Therefore we have,
for $\phi \in {\mathcal T} (T(\mathbb R))$ and $0 < R \leq 1$
\begin{align*}
\phi(0)& = \lim_{\epsilon \rightarrow +0}\int [K(x+i\epsilon
-i)+K(x-i\epsilon +i)]
\phi(x) dx\\
& = \lim_{\epsilon \rightarrow +0}\sum_{\omega=\pm 1}\int K(x-i\omega \epsilon +i\omega)\phi(x)dx\\
& = \lim_{\epsilon \rightarrow +0}\sum_{\omega=\pm 1}\int K(x-iR\omega -i\omega \epsilon +i\omega)\phi(x-iR\omega)dx\\
& = \sum_{\omega=\pm 1}\int K(x-i\omega R +i\omega)\phi(x-iR
\omega)dx
\end{align*}
Let $K_{r}(z) = r^{-1}K(z/r)$. Then we can reformulate the above
relation as
$$      \delta (x) = \lim_{\epsilon \rightarrow +0}[K_{r}(x+ir\epsilon -ir)+K_{r}(x-ir\epsilon +ir)]$$
and for $1 < R \leq  r$
$$      \sum _{\omega =\pm 1} \int _{-\infty }^{\infty } K_{r}(x+i(r-R)\omega ) \phi (x - iR\omega )dx = \phi (0)$$
and
$$      \sum _{\omega =\pm 1} \int _{-\infty }^{\infty } K_{r}(x - t +i(r-R)\omega ) \phi (x - iR\omega )dx$$
\begin{equation} \label{reproduce}
 = \sum _{\omega =\pm 1} \int _{-\infty }^{\infty } K_{r}(x
+i(r-R)\omega ) \phi (x + t - iR\omega )dx = \phi (t).
\end{equation}
Let $\Gamma _{1} = \{ y \in  \mathbb R; y > \ell \} $ and $\Gamma
_{2} = \{ y \in  \mathbb R; -y > \ell \} $. Given $F_{j} \in
{\mathcal O} _{0}(T(\Gamma _{j}))$ denote by  $u_{j}$ the tempered
ultra-hyperfunction defined by
$$      u_{j}(\phi ) = \int _{-\infty }^{\infty }F_{j}(x + iy)\phi (x + iy) dx$$
for $y \in  \Gamma _{j}$.  Choose $r > \ell $ and define
$$      U_{j}(z) = u_{j}*K_{r}(z) = \int _{-\infty }^{\infty }F_{j}(\xi  + i\eta )K_{r}(z - \xi  -i\eta ) d\xi $$
for $\eta  \in  \Gamma _{j}$.  Then $U_{j}(z)$ is analytic in
$V_j$,
\begin{align*}
  V_{1} &= \cup _{\eta \in \Gamma _{1}}\{ z \in  \mathbb C; \vert {\rm Im \, }(z - i\eta )\vert  < r\}  = \{ z \in  \mathbb C; {\rm Im \, }z > \ell  - r\} ,\\
 V_{2} &= \cup _{\eta \in \Gamma _{2}}\{ z \in  \mathbb C; \vert {\rm Im \, }(z - i\eta )\vert  < r\}  = \{ z \in  \mathbb C; {\rm Im \, }z < -\ell  + r\} .
\end{align*}
Note that (\ref{reproduce}) implies
$$      \sum _{\omega  = \pm 1} \int _{-\infty }^{\infty } U_j (x + i(r-R)\omega )\phi (x - iR\omega )dx$$
$$      = \sum _{\omega  = \pm 1} \int _{-\infty }^{\infty }u_j *K_{r}(x +i(r-R)\omega )\phi (x - iR\omega )dx$$
$$      = \sum _{\omega  = \pm 1} \int _{-\infty }^{\infty } \int _{-\infty }^{\infty } F_{j}(\xi  + i\eta )K_{r}(x - \xi - i\eta +i(r-R)\omega )\phi (x - iR\omega )dx d\xi $$
\beq \label{ukernel1}
     = \int _{-\infty }^{\infty }F_{j}(\xi +
i\eta )\phi (\xi + i\eta ) d\xi = u_j (\phi ). \eeq If $u_{1} =
u_{2}$, then $U_{1}(z) = U_{2}(z)$ in $V_{1} \cap V_{2} = \{ z \in
\mathbb C; \vert {\rm Im\, } z\vert < r - \ell \} $, and the two
functions $U_{1},U_2$ are continued to a function $U$ which is
analytic in $V_{1} \cup V_{2} = \mathbb{C}$.

Now introduce the function
$$  H(z) = \sum _{\omega =\pm 1} U(z + ir\omega );$$
clearly $H$ is an entire function, and we have for $y \in  \Gamma
_{1}$, $\phi  \in  {\mathcal T} (T(\mathbb R))$, and $\ell - r
\leq  \sigma  <  \ell $
$$      \int _{-\infty }^{\infty }H(x + iy)\phi (x + iy) dx = \sum _{\omega =\pm 1} \int _{-\infty }^{\infty }U(x + iy + ir\omega )\phi (x + iy) dx$$
$$      = \sum _{\omega =\pm 1} \int _{-\infty }^{\infty }U(x + i(r - \ell +\sigma )\omega )\phi (x - i(\ell - \sigma )\omega ) dx$$
$$      = \sum _{\omega =\pm 1} \int _{-\infty }^{\infty }U_{1}(x + i(r - \ell +\sigma )\omega )\phi (x - i(\ell - \sigma )\omega ) dx$$
$$      = u_{1}(\phi ) = \int _{-\infty }^{\infty }F_{1}(x + iy)\phi (x + iy) dx$$
 where we used the fact that
$U(z) = U_{1}(z)$ for ${\rm Im\, } z > \ell  - r$ and relation
(\ref{ukernel1}), hence
$$      \int _{-\infty }^{\infty }H(x + iy)\phi (x + iy) dx = \int _{-\infty }^{\infty }F_{1}(x + iy)\phi (x + iy) dx$$
for $\phi  \in  {\mathcal T} (T(\mathbb R))$.  This shows that
$H(z) = F_{1}(z)$ in $T(\Gamma _{1})$, and in the same way we get
$H(z) = F_{2}(z)$ in $T(\Gamma _{2})$.  This proves the global
form of the EOW theorem for tempered ultra-hyperfunctions of one
variable.

 Next we discuss the local version. To this end assume that the carrier of $u_{1} - u_{2}$ is contained in $L =[a, b] + i[-\ell , \ell ]$, instead of $u_{1} = u_{2}$. This will lead to the local form of the EOW theorem.

Introduce the function $U_{12}$ by
$$U_{12}(z) = (u_{1} - u_{2})*K_{r}(z).$$
Since the singular points of $K_{r}(z)$ are $z = i(1 + 2n)r$, $n =
0, \pm 1, \pm 2, \ldots $, $K_{r}(z - w)$, for ${\rm Re \, }z
\not\in [a, b]$, is holomorphic in a neighborhood of $L$ and
$U_{12}(z)$ is holomorphic in $Z = \{ z \in  \mathbb C; {\rm Re \,
}z \not\in [a, b]\} $.  Since $u_{1} = u_{2} + u_{1} - u_{2}$,
$$      U_{1}(z) = U_{2}(z) + U_{12}(z).$$
$U_{1}(z)$ is holomorphic in $V_{1}$ and $U_{2}(z) + U_{12}(z)$ is
holomorphic in $Z\cap V_{2}$.  Therefore, in the same way as in
the global case, $U_{1}(z)$ can be analytically continued to a
function $U_{1}^{\prime }(z)$ which is holomorphic in $Z$ and
$F_{1}(z)$ can be analytically continued to $H_{1}(z) = \sum
_{\omega =\pm 1} U_{1}^{\prime }(z + ir\omega )$ which is
holomorphic in $Z$.
 In the same way, $F_{2}(z)$ can
be analytically continued to $H_{2}(z)$ which is holomorphic in
$Z$. In order to show $H_{1}(z) = H_{2}(z)$ in $Z$, we introduce a
path $C = C_{1} + \ldots + C_{5}$ consisting of line segments
$C_{1} = (-\infty , a - 2\ell ]$, $C_{2} = [a - 2\ell , a+i2\ell
]$, $C_{3} = [a+i2\ell , b+i2\ell ]$, $C_{4} = [b+i2\ell , b +
2\ell ]$ and $C_{5} = [b+2\ell , \infty )$ (see Figure 1).
\begin{figure}[htbp]
  \begin{center}
   \includegraphics[width=0.8\textwidth]{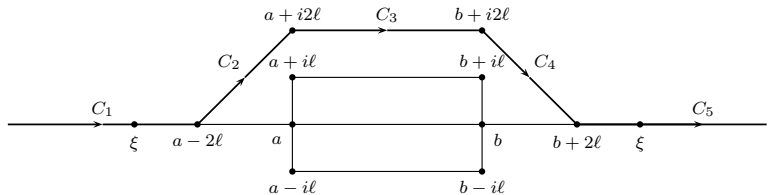}
    \caption{The path of integration $C=C_1 + C_2 +C_3 +C_4+C_5$}
    \label{fig:path1}
  \end{center}
\end{figure}
\noindent Let $E_{\xi }^{t}(z) = (4\pi t)^{-1/2} e^{-(\xi  -
z)^{2}/4t}$ and $\xi  \in  \mathbb R$ such that $\xi  < a-2\ell $
or $\xi
> b + 2\ell $.  Then we have
$$ u_{1}(E_{\xi }^{t}) = \int_{-\infty }^{\infty }F_{1}(x + i2\ell )E_{\xi }^{t}(x+i2\ell ) dx$$
$$ = \int_{-\infty }^{\infty }H_{1}(x +i2\ell )E_{\xi }^{t}(x+i2\ell ) dx = \int _{C}H_{1}(z)E_{\xi }^{t}(z) dz \rightarrow  H_{1}(\xi )$$
as $t \rightarrow  0+$.  Here we used the fact that if $\vert \xi
- {\rm Re \, }z\vert  > \vert {\rm Im \, }z\vert $, $\vert
E^{t}_{\xi }(z)\vert  = (4\pi t)^{-1}e^{-(\vert \xi  - {\rm Re \,
}z\vert ^{2} - \vert {\rm Im \, }z\vert ^{2})/4t} \rightarrow  0$
as $t \rightarrow  0+$ and if $z \in  \mathbb R $, $E^{t}_{\xi
}(z) \rightarrow  \delta (z - \xi )$, i.e., $E_{\xi }^{t}(z)
\rightarrow  0$ for $z \in  C_{2} + C_{3} + C_{4}$, and if $\xi
\in  C_{1}$ (resp. $\xi  \in  C_{5}$) then $E_{\xi }^{t}(z)
\rightarrow  \delta (z-\xi )$ and $E_{\xi }^{t}(z) \rightarrow  0$
for $z \in  C_{5}$ (resp. $z \in C_{1}$).  If we choose a curve
$C^{\prime } = C_{1}^{\prime } + \ldots + C_{5}^{\prime }$, where
$C_{1}^{\prime } = (-\infty , a - 2\ell ]$, $C_{2}^{\prime } = [a
- 2\ell , a-i2\ell ]$, $C_{3}^{\prime } = [a-i2\ell , b-i2\ell ]$,
$C_{4}^{\prime } = [b-i2\ell , b + 2\ell ]$, $C_{5}^{\prime } =
[b+2\ell , \infty )$, then we have
$$ u_{2}(E_{\xi }^{t}) = \int _{-\infty }^{\infty }H_{2}(x -i2\ell )E_{\xi }^{t}(x-i2\ell ) dx = \int _{C^{\prime }}H_{2}(z)E_{\xi }^{t}(z) dz \rightarrow  H_{2}(\xi )$$
as $t \rightarrow  0+$.  Moreover, we find
$$\vert u_{1}(E_{\xi }^{t}) - u_{2}(E_{\xi }^{t})\vert  \leq  M_{V} \sup_{z\in V}\vert E_{\xi }^{t}(z)\vert  \rightarrow  0$$
as $t \rightarrow  0+$ for some neighborhood $V$ of $L$.  Thus
 $H_{1}(\xi ) = H_{2}(\xi )$ follows and consequently there exists a
function $H$ which is holomorphic in $\mathbb C\backslash L$ and
which is the common extension of $F_{1}$ and $F_{2}$.  Therefore
$F_{1}$ and $F_{2}$ have the common extension $H \in {\mathcal O}
_{0}(\mathbb C\backslash L)$, and $H$ is an element of
$H^{1}_{L}(\mathbb C, {\mathcal O} _{0})$ of (\ref{cohomology-u}).

It is interesting to note that the cohomology group
$H^{1}_{L}(\mathbb C, {\mathcal O} _{0})$ of tempered
ultra-hyperfunctions with carriers in $L$ appears again here in
the argument of EOW theorem using the kernel $K_r(z)$.

\section{Global edge of the wedge theorem}
In this section we prove the global version of EOW theorem {in higher dimension} using
the functional method. This will help in understanding the proof
of the local version.

First we recall the test-function space ${\mathcal T} (T(\mathbb
R^{n}))$ of tempered ultra-hyperfunction. Let $T(A) = \mathbb
R^{n} + iA$ for $A \subset  \mathbb R^{n}$. Let $K \subset \mathbb
R^{n}$ be a convex compact set and ${\mathcal T} _{b}(T(K))$ the
set of continuous functions on $T(K)$ which are holomorphic in the
interior of $T(K)$ and satisfy
$$      \Vert f\Vert ^{T(K), j} = \sup \{ \vert z^{p}f(z)\vert ; z \in  T(K), \vert p\vert  \leq  j\}  < \infty , \  j = 1,2, \ldots\; .$$
There is a natural restriction mapping $P_{K,L}: {\mathcal T}
_{b}(T(K)) \rightarrow  {\mathcal T} _{b}(T(L))$ for $K \supset
L$.
\begin{defn} ${\mathcal T} (T(\mathbb R^{n}))$ is the projective
limit
$$      {\mathcal T} (T(\mathbb R^{n})) = \lim _{\leftarrow } {\mathcal T} _{b}(T(K)), \  K \uparrow  \mathbb R^{n}$$
of the projective system $({\mathcal T} _{b}(T(K)), P_{K,L})$,
where $K$ runs through the convex compact sets in $\mathbb R^{n}$.
\end{defn}
It is known that the kernel $K(z)$ of (\ref{hkernel}) is a rapidly
decreasing holomorphic function in $\{ z\in \mathbb C^{n}; \vert
{\rm Im \, }z\vert ^{2} < 1 + \vert {\rm Re \, }z\vert ^{2}\} $
(see [\cite{NN93}]), and  the following lemma holds (see
[\cite{NN01}]).
\begin{lem}  For any $0 < R \leq  1$ and $\phi  \in {\mathcal T}(T(\mathbb R^{n}))$ one has
$$      \phi (t) = \int _{\vert \omega \vert =1}d\omega  \int _{ \mathbb R^{n}} K(x-iR\omega +i\omega -t)\phi (x-iR\omega )dx.$$
\end{lem}
\begin{proof}  Consider the Fourier transform $\hat{\phi }$ of $\phi$ and note that it can be represented as
$$ \hat{\phi} (\xi ) = \int _{ \mathbb R^{n}} \phi (x)e^{i\langle x, \xi \rangle } dx = \int _{ \mathbb R^{n}} \phi (x - iR\omega )e^{i\langle x - iR\omega , \xi \rangle } dx.$$
Thus we get
$$      \int _{\vert \omega \vert =1} d\omega  \int _{ \mathbb R^{n}} K(x - iR\omega  + i\omega  - t)\phi (x - iR\omega ) dx$$
$$      = (2\pi )^{-n} \int _{\vert \omega \vert =1} d\omega  \int _{ \mathbb R^{n}} dx \phi (x - iR\omega )\int d\xi e^{i\langle x - iR\omega  + i\omega  - t, \xi \rangle }/I(\xi ) $$
$$      = (2\pi )^{-n} \int _{\vert \omega \vert =1} d\omega  \int d\xi  \int _{ \mathbb R^{n}} dx \phi (x - iR\omega )e^{i\langle x - iR\omega , \xi \rangle } e^{i\langle i\omega  - t, \xi \rangle }/I(\xi ) $$
$$      = (2\pi )^{-n} \int _{\vert \omega \vert =1} d\omega  \int d\xi  \hat{\phi } (\xi ) e^{i\langle i\omega  - t, \xi \rangle }/I(\xi ) $$
$$      = (2\pi )^{-n}  \int d\xi  \int _{\vert \omega \vert =1} d\omega
\hat{\phi } (\xi ) e^{i\langle i\omega  - t, \xi \rangle }/I(\xi )
$$
$$      = (2\pi )^{-n}  \int d\xi  \hat{\phi } (\xi ) e^{-i\langle t, \xi \rangle }I(\xi )/I(\xi ) = (2\pi )^{-n}  \int d\xi  \hat{\phi } (\xi ) e^{-i\langle t, \xi \rangle } = \phi (t).$$
\end{proof}
\begin{cor} \label{corkr}
Let $K_{r}(z) = r^{-n} K(z/r)$. Then $K_{r}(z)$ is a rapidly
decreasing holomorphic function in $\{ z\in  \mathbb C^{n}; \vert
{\rm Im \, }z/r\vert ^{2} < 1 + \vert {\rm Re \, }z/r\vert ^{2}\}
$ = $\{ z\in  \mathbb C^{n}; \vert {\rm Im \, }z\vert ^{2} < r^{2}
+ \vert {\rm Re\, } z\vert ^{2}\} $, and for $0 < R \leq  r$, the
identity

$$      \phi (t) = \int _{\vert \omega \vert =1}d\omega  \int _{ \mathbb R^{n}} K_r (x-iR\omega +i\omega -t)\phi (x-iR\omega )dx$$
holds.
\end{cor}
\begin{proof} A straightforward calculation yields
$$      \int _{\vert \omega \vert =1}d\omega  \int _{ \mathbb R^{n}} K_{r}(x-iR\omega +ir\omega -t)\phi (x-iR\omega )dx$$
$$      = r^{-n} \int _{\vert \omega \vert =1}d\omega  \int _{ \mathbb R^{n}} K_{1}((x-iR\omega +ir\omega -t)/r)\phi (x-iR\omega )dx$$
$$      = \int _{\vert \omega \vert =1}d\omega  \int _{ \mathbb R^{n}} K_{1}(y-iR\omega /r+i\omega -t/r)\phi (r(y-iR\omega
/r))dy$$
$$   = \phi (r(t/r)) = \phi (t).$$
\end{proof}
Thus we can remove the restriction $R \leq  1$.
\begin{defn} \label{pincrease} Let $O$ be an open set in $\mathbb
R^{n}$. Denote by ${\mathcal O} _{0}(T(O))$  the set of those
functions $F(z)$ which are holomorphic in $T(O)$ and which satisfy
the following condition:
\newline
For any compact set $K \subset  O$, there exists a natural number
$j > 0$ such that
$$      \sup _{z \in  T(K)} \vert F(z)\vert (1 + \vert z\vert )^{-j} < \infty .$$
\end{defn}
Now we can state the global form of the edge of the wedge theorem
for tempered ultra-hyperfunctions.

\begin{thm}  \label{GEOW} Let $V_{+} = \{ y \in  \mathbb R^{n};
y_{1}
> \sqrt{\sum _{j=2}^{n} y_{j}^{2}}\} $ be the forward light-cone
in $\mathbb R^{n}$, $e = (1, 0, \ldots , 0) \in  \mathbb R^{n}$
and $\Gamma  = \ell e + V_{+}$.  Let $F_{1}(z) \in  {\mathcal O}
_{0}(T(\Gamma ))$, $F_{2}(z) \in  {\mathcal O} _{0}(T(-\Gamma ))$
and $C^{r}_{j} = \{ z_{j} \in \mathbb C; z_{j} = x + ir, \ -\infty
< x < \infty \} $.  Define
$$      u_{1}(\phi ) = \int _{ \prod _{j} C^{\eta _{j}}_{j}} F_{1}(z)\phi (z) dz,$$
$$      u_{2}(\phi ) =\int _{ \prod _{j} C^{-\eta _{j}}_{j}} F_{2}(z)\phi (z) dz$$
for $\phi  \in  {\mathcal T} (T(\mathbb R^{n}))$ and $\eta  =
(\eta _{1}, \ldots , \eta _{n}) \in  \Gamma $.  If $u_{1} =
u_{2}$, then $F_{1}(z)$ and $F_{2}(z)$ can be continued
analytically to each other and define an entire function.
\end{thm}
\begin{proof}
Note first that $u_{1}(\phi )$ (resp. $u_{2}(\phi )$) does not
depend on the path $C^{\eta _{j}}_{j}$ (resp. $C^{-\eta
_{j}}_{j}$).  This can be proved by applying Cauchy's integral
theorem $n$ times. Let $r > \ell $.  Observe that the function
$U_1$ defined by
$$      U_{1}(z) = u_{1}*K_{r}(z) = \int _{ \mathbb R^{n}} F_{1}(\xi  + i\eta )K_{r}(z - \xi  - i\eta ) d\xi, $$
for $\eta  \in  \Gamma $, is analytic in
$$      V_{1} = \cup _{\eta \in \Gamma } \{ z\in  \mathbb C^{n}; \vert {\rm Im \, }(z - i\eta )\vert  < r\}  = \{ z \in  \mathbb C^{n}; {\rm dist \, }({\rm Im \, }z , \Gamma ) < r\} $$
$$      \supset  \{ z \in  \mathbb C^{n}; {\rm dist \, }({\rm Im \, }z , \ell e) < r\}  \supset  \{ z \in  \mathbb C^{n}; \vert {\rm Im \, }z \vert  < r - \ell \} $$
and similarly $U_{2}(z) = u_{2}*K_{r}(z)$ is analytic in
$$      V_{2} = \cup _{\eta \in -\Gamma } \{ z\in  \mathbb C^{n}; \vert {\rm Im \, }(z - i\eta )\vert  < r\}  =  \{ z \in  \mathbb C^{n}; {\rm dist \, }({\rm Im \, }z , -\Gamma ) < r\} $$
$$      \supset  \{ z \in  \mathbb C^{n}; {\rm dist \, }({\rm Im \, }z , -\ell e) < r\}  \supset  \{ z \in  \mathbb C^{n}; \vert {\rm Im \, }z \vert  < r - \ell \} .$$
 Corollary \ref{corkr} implies
\beq \label{ukernel}
      \int _{\vert \omega \vert =1} d\omega  \int dx U_{j}(x +
i(r - R)\omega ) \phi (x - iR\omega ) \eeq
$$      = \int _{\vert \omega \vert =1} d\omega  \int dx u_{j}*K_{r}(x + i(r - R)\omega ) \phi (x - iR\omega )$$
$$      = \int _{\vert \omega \vert =1} d\omega  \int dx \int d\xi  F_{j}(\xi  + i\eta )K_{r}(x + i(r - R)\omega  - \xi  - i\eta ) \phi (x - iR\omega )$$
$$      = \int d\xi  \int _{\vert \omega \vert =1} d\omega  \int dx F_{j}(\xi  + i\eta )K_{r}(x + i(r - R)\omega  - \xi  - i\eta ) \phi (x - iR\omega )$$
$$      = \int d\xi  F_{j}(\xi  + i\eta )\phi (\xi  + i\eta ) = u_{j}(\phi ).$$
The relation $u_{1}(\phi ) = u_{2}(\phi )$ implies $U_{1}(z) =
U_{2}(z)$ in $V_{1} \cap V_{2} \supset \{ z \in  \mathbb C^{n};
\vert {\rm Im \, }z \vert  < r - \ell \} $.  Then $U_{1}(z)$ and
$U_{2}(z)$ are continued to an analytic function $U(z)$ in $V_{1}
\cup V_{2}$.  Moreover, $U(z)$ is analytically continued to the
convex envelope of $V_{1} \cup V_{2}$ by Bochner's theorem on
tublar domains (see [\cite{Vl64}]).  Since the convex envelope of
$V_{1} \cup V_{2}$ is entire space $\mathbb C^{n}$, $U(z)$ is
analytically continued to an entire function.  The function $H$,
$$      H(z) = \int _{\vert \omega \vert =1}d\omega  U(z + ir\omega ),$$
is an entire function which satisfies
$$  \int_{ \prod _{j} C^{\eta _{j}}} H(z)\phi (z) dz = \int _{ \prod _{j} C^{\eta _{j}}} dz\int _{\vert \omega \vert
=1}d\omega  U(z + ir\omega )\phi (z)$$
$$      = \int _{\vert \omega \vert =1}d\omega  \int _{ \mathbb R^{n}} dx U(x + i\eta  + ir\omega )\phi (x + i\eta )$$
$$      = \int _{\vert \omega \vert =1}d\omega  \int _{ \mathbb R^{n}} dx U(x + i(r - [\ell  + \sigma ])\omega )\phi (x - i[\ell  + \sigma ]\omega )$$
$$      = \int _{\vert \omega \vert =1}d\omega  \int _{ \mathbb R^{n}} dx U_{1}(x + i(r - [\ell  + \sigma ])\omega )\phi (x - i[\ell  + \sigma ])\omega )$$
$$      = u_{1}(\phi ) = \int _{ \prod _{j} C^{\eta _{j}}} F_{1}(z)\phi (z) dz$$
for $0 < \sigma  \leq  r - \ell $, where we used the fact that
$U(z) = U_{1}(z)$ in $\{ z \in \mathbb C^{n}; \vert {\rm Im \, }z
\vert < r - \ell \} $ and Relation (\ref{ukernel}). Thus we have
$$      \int _{ \prod _{j} C^{\eta _{j}}} H(z)\phi (z) dz = \int _{ \prod _{j} C^{\eta _{j}}} F_{1}(z)\phi (z) dz$$
for $\phi  \in  {\mathcal T} (T(\mathbb R^{n}))$.  This shows that
$H(z) = F_{1}(z)$ in $T(\Gamma )$, and in the same way we get
$H(z) = F_{2}(z)$ in $T(-\Gamma )$.  Thus $F_{1}(z)$ and
$F_{2}(z)$ are analytically continued to an entire function
$H(z)$.  This completes the proof.
\end{proof}

\section{Local edge of the wedge theorem}
We begin by proving a regularization result for tempered
ultrahyperfunctions using the kernel $K_r$.
\begin{prop}  \label{Zanalytic} Let $L$ be an open set in $\{ w \in
\mathbb C^{n}; \vert {\rm Im \, }w\vert  < r\} $ and $u \in
{\mathcal T} (L)^{\prime }$.  Then $U(z) = K_{r}*u(z)$ is
holomorphic in
$$      Z = \{ z \in  \mathbb C^{n}; \vert {\rm Im \, }(z - w)\vert ^{2} < r^{2} + \vert {\rm Re \, }(z - w)\vert ^{2}, \  \forall w \in  L\} .$$
 Furthermore, introduce the function \beq \label{eq:gr}
    g_r(x) = \inf \{ \sqrt{r^{2} + \vert x - {\rm Re \, }w\vert
^{2}} - \vert {\rm Im \, }w\vert ; w \in  L\}. \eeq Then the
following inclusion is valid. \beq \label{eq:imaginary}
     Z \supset  \{ z = x + iy \in  \mathbb C^{n}; \vert y\vert <
g_r(x), x \in  \mathbb R^{n}\} . \eeq

\end{prop}
\begin{proof}
Since $K_{r}(z)$ is a rapidly decreasing holomorphic function in
$\{ z\in \mathbb C^{n}; \vert {\rm Im \, }z\vert ^{2} < r^{2} +
\vert {\rm Re\, } z\vert ^{2}\} $, if $z \in Z$, $K_{r}(z - w)$ is
a rapidly decreasing holomorphic function of $w$ in a neighborhood
of $L$. The inclusion (\ref{eq:imaginary}) can be shown as
follows:
$$      Z = \{ z \in  \mathbb C^{n}; \vert {\rm Im \, }(z - w)\vert  < \sqrt{r^{2} + \vert {\rm Re \, }(z - w)\vert ^{2}}, \  \forall w \in  L\} $$
$$      \supset  \{ z \in  \mathbb C^{n}; \vert {\rm Im \, }z\vert  + \vert {\rm Im \, }w\vert  < \sqrt{r^{2} + \vert {\rm Re \, }(z - w)\vert ^{2}}, \  \forall w \in  L\} $$
$$    \supset  \{ z = x + iy \in  \mathbb C^{n}; \vert y\vert <
g_r(x), x \in  \mathbb R^{n}\}. $$

\end{proof}
\begin{thm}  \label{lwedge}  Let $F_{i}(z)$ $(i = 1, 2)$ be holomorphic functions and $u_{i}$ be the tempered ultra-hyperfunctions defined by $F_{i}(z)$ as in Theorem \ref{GEOW}. Let $L$ be an open set in $\{ w \in \mathbb C^{n}; \vert {\rm Im
\, }w \vert < \ell \} $ such that the set $O = \{ x \in \mathbb
R^{n}; g_r(x) > r\} $ contains an open set $Q$ such that ${\rm
dist \, }(\partial O, Q)
> 2\ell $,
 for $r > \ell /(\sqrt{2} - 1)$ and $g_r(x)$ of
(\ref{eq:gr}).
 Assume that $u_{1} - u_{2} \in {\mathcal T}
(L)^{\prime }$.  Then $F_{i}(z)$ are analytically continued to $O$
and coincide there.
\end{thm}
\begin{proof}  We use the same notation as in Theorem \ref{GEOW}.  Then $U_{1}(z) =
u_{1}*K_{r}(z)$ is analytic in
$$      V_{1} = \{ z \in  \mathbb C^{n}; {\rm dist \, }({\rm Im \, }z , \Gamma ) < r\} $$
and $U_{2}(z) = u_{2}*K_{r}(z)$ is analytic in
$$      V_{2} = \{ z \in  \mathbb C^{n}; {\rm dist \, }({\rm Im \, }z , -\Gamma ) < r\} .$$
Since $u_{1} - u_{2} \in  {\mathcal T} (L)^{\prime }$, it follows
from Proposition \ref{Zanalytic} that $U_{12}(z) = (u_{1} -
u_{2})*K_{r}(z)$ is analytic in
$$  \supset  \{ z = x + iy \in  \mathbb C^{n}; \vert y\vert <
g_r(x), x \in  \mathbb R^{n}\} .$$
 Since $u_{1}*K_{r}(z) =
u_{2}*K_{r}(z) + (u_{1} - u_{2})*K_{r}(z)$, $U_{1}(z)$ is analytic
in
$$      V_{1} \cup (V_{2} \cap \{ z = x + iy \in  \mathbb C^{n}; \vert y\vert  < g_r(x), x \in  \mathbb R^{n}\} )$$
$$      \supset   \{ x \in  \mathbb R^{n}; g_r(x) > r + \delta \} $$
$$      \times  i\left( \{ y \in  \mathbb R^{n}; {\rm dist \, }(y , \Gamma ) < r\}
  \cup (\{ y \in  \mathbb R^{n}; {\rm dist \, }(y , - \Gamma ) < r\} \cap B_{r + \delta } )\right)$$
$$     =  \{ x \in  \mathbb R^{n}; g_r(x) > r + \delta \} \times  i\left( \{ y \in  \mathbb R^{n}; {\rm dist \, }(y ,
\Gamma ) < r\}
  \cup B_{r + \delta } \right)$$
for any $r + \delta  < \sqrt{2}r + \ell $, where $B_r = \{ y \in
\mathbb R^{n}; \vert y\vert  < r\} $ (see Figure 2).  Note that if
$r > \ell /(\sqrt{2} - 1)$ then there exists $\delta > 0$ such
that $r + \delta  < \sqrt{2}r + \ell $.
\begin{center}
\begin{figure}[h!]
\includegraphics[width=0.9\textwidth]{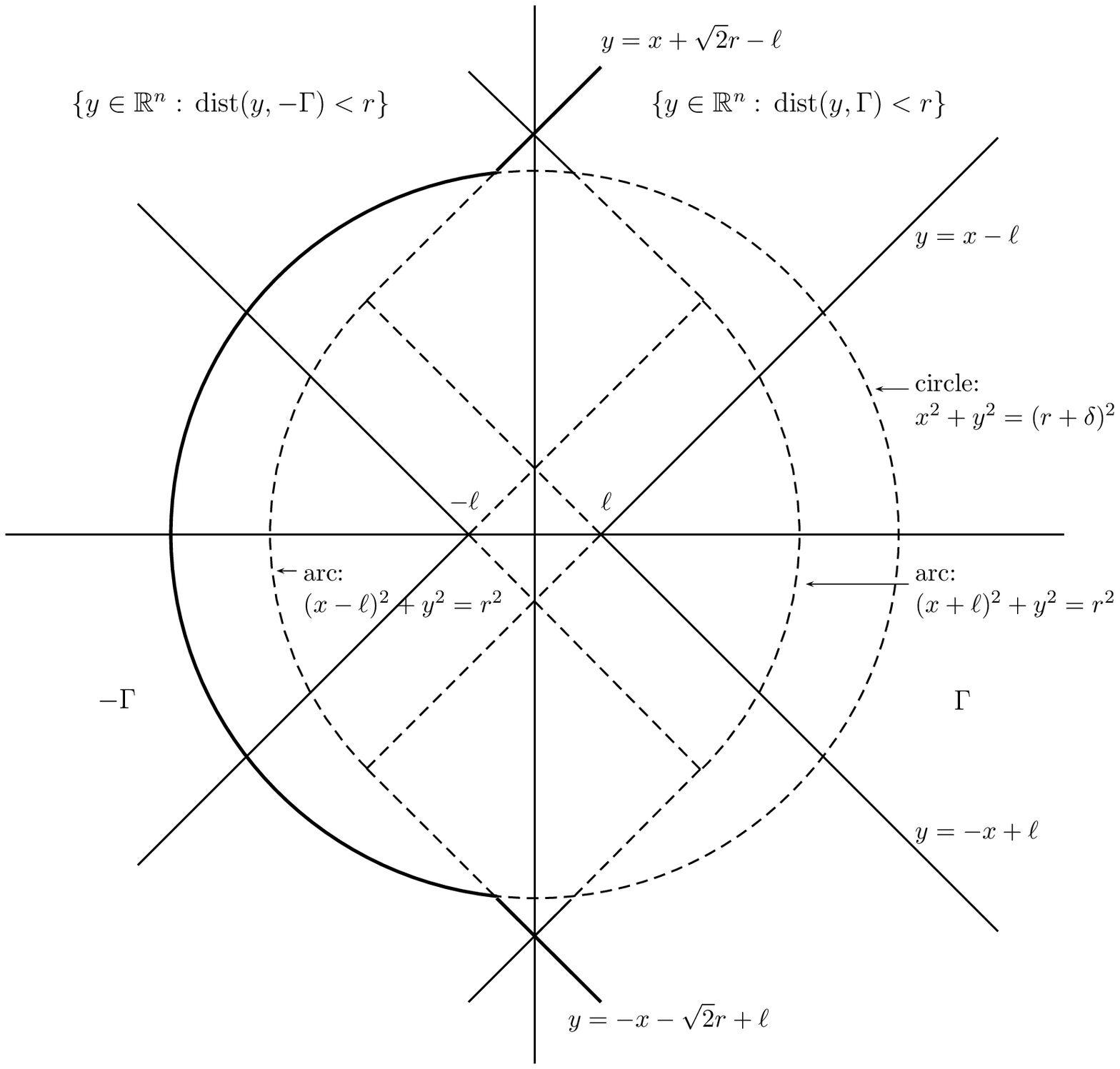}
\caption{$\{ y \in  \mathbb R^{n}; {\rm dist \, }(y , \Gamma ) <
r\}
  \cup (\{ y \in  \mathbb R^{n}; {\rm dist \, }(y , - \Gamma ) < r\} \cap B_{r + \delta }) = \{ y \in  \mathbb R^{n}; {\rm dist \, }(y , \Gamma ) < r\}
  \cup B_{r + \delta } $}
  \end{figure}
\end{center}
Denote
$$ W_{r, \delta } = \{ y \in  \mathbb R^{n}; {\rm dist \, }(y , \Gamma ) < r \} \cup B_{r + \delta } .$$
If
$$      r + \delta > \sqrt{(r/\sqrt{2})^{2} + (r/\sqrt{2} - \ell )^{2}}$$
then $\cap _{\vert \omega \vert =1} (W_{r, \delta } + r\omega )$
strictly contains the set  $\Gamma $, and if $\delta
> 0$ then $ \cap _{\vert \omega \vert =1} (W_{r, \delta } + r\omega )$
contains a connected set $\tilde{\Gamma } _{\delta }$ containing
$\Gamma $ and an open ball $B_{\delta } $ (see Figure 3).
\begin{center}
\begin{figure}[h!]
\includegraphics[width=0.9\textwidth]{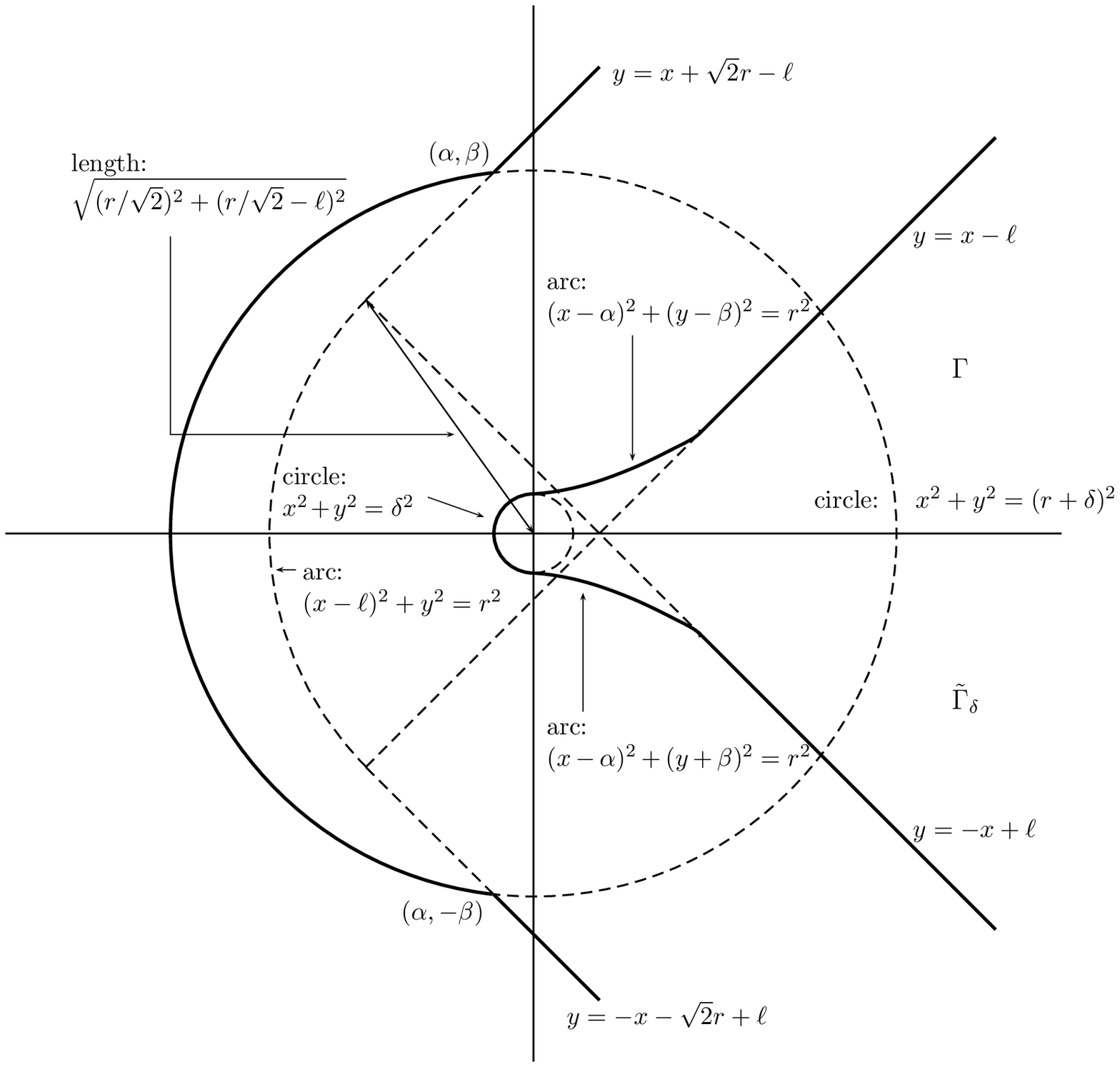}
\caption{The connected set $\tilde{\Gamma}_{\delta }$ containing
$\Gamma $ and $B_{\delta } $}
\end{figure}
\end{center}
As before, consider
$$ H_{1}(z) = \int _{\vert \omega \vert =1}d\omega  U_{1}(z + ir\omega ).$$
 $H_{1}(z)$ is holomorphic in the set
$$ \cap _{\vert \omega \vert =1} [\{ V_{1} \cup (V_{2} \cap \{ z = x + iy \in  \mathbb C^{n}; \vert y\vert  < g_r(x), x \in  \mathbb R^{n}\} )\}  + ir\omega ]$$
which contains the sets $T(\Gamma )$ and
$$ \{ x \in  \mathbb R^{n}; g_r(x) > r + \delta \}  \times  i \cap _{\vert \omega \vert =1} (W_{r + \delta } + r\omega )  \supset  \{ x \in  \mathbb R^{n};g_r(x) > r + \delta \}
\times  i \tilde{\Gamma } _{\delta }.$$ This implies
$$      \int _{ \prod _{j} C^{\eta _{j}}} H_{1}(z)\phi (z) dz = \int _{ \prod _{j} C^{\eta _{j}}} dz\int _{\vert \omega \vert =1}d\omega  U_{1}(z + ir\omega )\phi (z)$$
$$      = \int _{\vert \omega \vert =1}d\omega  \int _{ \mathbb R^{n}} dy U_{1}(y + i\eta  + ir\omega )\phi (y + i\eta )$$
$$      = \int _{\vert \omega \vert =1}d\omega  \int _{ \mathbb R^{n}} dy U_{1}(y + i(r - [\ell + \sigma ])\omega )\phi (y - i[\ell  + \sigma ]\omega )$$
$$      = u_{1}(\phi ) = \int _{ \prod _{j} C^{\eta _{j}}} F_{1}(z)\phi (z) dz.$$
This shows
$$  \int _{ \prod _{j} C^{\eta _{j}}} H_{1}(z)\phi (z) dz = \int _{ \prod _{j} C^{\eta _{j}}} F_{1}(z)\phi (z) dz$$
for $\phi  \in  {\mathcal T} (T(\mathbb R^{n}))$ and therefore
$H_{1}(z) = F_{1}(z)$ in $T(\Gamma )$ and $F_{1}(z)$ is
analytically continued to
$$    \{ x \in  \mathbb R^{n}; g_r(x) > r + \delta \}  + i\tilde{\Gamma }_{\delta }.$$
In the same way we get $H_{2}(z) = F_{2}(z)$ in $T(- \Gamma )$ and
$F_{2}(z)$ is analytically continued to
$$ \{ x \in  \mathbb R^{n}; g_r(x) > r + \delta \}  - i\tilde{\Gamma }_{\delta }.$$
In order to show $H_{1}(z) = H_{2}(z)$, we choose a surface
$S_{1}$ contained in the domain of the holomorphy of $H_{1}(z)$
such that
$$      S_{1} = \{ z = x + iy \in  \mathbb C^{n}; y_{1} = f_{1}(x), y_{j} = 0, 2 \leq  j \leq  n, x \in  \mathbb R^{n}\} ,$$
where $0 \leq f_{1}(x) \leq \ell + \delta $ is a continuous
function satisfying $f_{1}(x) = 0$ for $x \in  O = \{ x \in
\mathbb R^{n}; g_r(x) > r\}$ and $f_{1}(x) = \ell + \delta $ for
$$      x \in  \{ x \in  \mathbb R^{n}; g_r(x) < \sqrt{(r/\sqrt{2})^{2} + (r/\sqrt{2} - \ell )^{2}}\} .$$
For $\xi  \in  Q \subset O$ define  $E^{t}_{\xi }(z) = (4\pi
t)^{-n/2} e^{-(\xi  - z)^{2}/4t}$. Note that the open ball with
center $\xi $ and radius $2\ell $ is contained in $O$, $B_{2\ell
}(\xi ) \subset O \subset S_{1}$. Then, for $\xi  \in  Q$,
$$ u_{1}(E_{\xi }) = \int _{S_{1}} H_{1}(z)E_{\xi }^{t}(z)dz$$
$$= \int_{B_{2\ell}(\xi)} H_{1}(z)E(\xi-z)dz + \int_{S_{1}\backslash B_{2\ell}(\xi)}H_{1}(z)E_{\xi }^{t}(z)dz$$
and
$$ \int_{B_{2\ell }(\xi)} H_{1}(z)E_{\xi}^{t}(z)dz = \int_{B_{2\ell}(\xi)} H_{1}(z)(4\pi t)^{-n/2} e^{-(\xi  - x)^{2}/4t} dx \rightarrow  H_{1}(\xi )$$
as $t \rightarrow  0+$.  Since
$$      \int _{S_{1} \backslash B_{2\ell }(\xi )}H_{1}(z)E_{\xi }^{t}(z)dz = \int _{S_{1} \backslash B_{2\ell }(\xi )}H_{1}(z)(4\pi t)^{-n/2} e^{-(\xi  - z)^{2}/4t}dz$$
and for $z \in  S_{1} \backslash  B_{2\ell }(\xi )$,
$$      \vert e^{-(\xi  - z)^{2}/4t}\vert  = e^{-(\xi  - x)^{2}/4t}  e^{y_{1}^{2}/4t} \vert e^{-2i(\xi _{1} - x_{1})y_{1}/4t}\vert \leq  e^{-(\xi  - x)^{2}/4t} e^{(\ell  + \delta )^{2}/4t},$$
we get
$$      \int _{S_{1} \backslash B_{2\ell }(\xi )}H_{1}(z)E_{\xi }^{t}(z)dz \rightarrow  0$$
as $t \rightarrow  0+$, where we used the relation $(\xi  - x)^{2}
\geq  4\ell ^{2}$ and consequently, $(\xi  - x)^{2}/2 \geq  (\ell
+ \delta )^{2}$ for $0 < \delta  < \ell /2$.  It follows
$$      \lim _{t \rightarrow  0} u_{1}(E_{\xi }^{t}) = H_{1}(\xi ).$$
In the same way we get
$$      \lim _{t \rightarrow  0} u_{2}(E_{\xi }^{t}) = H_{2}(\xi ).$$
Let $x \in  O = \{ x \in  \mathbb R^{n}; g_r(x) > r\} $.  Then
$$      \sqrt{r^{2} + \vert x - {\rm Re \, }w\vert ^{2}} - \vert {\rm Im \, }w\vert  > r$$
for any $w \in  L$, that is, $\vert x - {\rm Re \, }w\vert  > 0$.
Let $\xi  \in  Q$.  Then $\vert \xi - {\rm Re \, }w\vert  > 2\ell
$ for any $w \in  L$ and
$$      \sup _{w \in  L} (1 + \vert w\vert )^{j} (4\pi t)^{-n/2} \exp [-(\xi  - {\rm Re \, }w)^{2} + ({\rm Im \, }w)^{2}]/4t$$
$$      \leq  \sup _{w \in  L} (1 + \vert w\vert )^{j} (4\pi t)^{-n/2} \exp [-(\xi  - {\rm Re \, }w)^{2}/2 - \ell ^{2}]/4t \rightarrow  0$$
as $t \rightarrow  0+$.  $u_{1} - u_{2} \in  {\mathcal T}
(L)^{\prime }$ implies
$$      0 = \lim _{t\rightarrow 0}(u_{1} - u_{2})(E^{t}_{\xi }) = \lim _{t\rightarrow 0}u_{1}(E^{t}_{\xi }) - \lim _{t\rightarrow 0}u_{2}(E^{t}_{\xi }) = H_{1}(\xi ) - H_{2}(\xi ).$$
Thus $H_{1}(z)$ and $H_{2}(z)$ are analytically continued to each
other. This completes the proof. \end{proof} The following
corollary is used in [\cite{NB09}].
\begin{cor}  \label{application}  Let $V_{+} = \{ y \in  \mathbb R^{4}; y_{0}
> \sqrt{\sum _{j=1}^{3} y_{j}^{2}}\} $ be the forward light-cone
in $\mathbb R^{4}$, $e = (1, 0, 0, 0) \in  \mathbb R^{4}$ and
$\Gamma = \ell e + V_{+}$.  Let $F_{1}(z) \in {\mathcal O}
_{0}(T(\Gamma ))$, $F_{2}(z) \in  {\mathcal O} _{0}(T(-\Gamma ))$
and $C^{r}_{j} = \{ z_{j} \in  \mathbb C; z_{j} = x + ir, \
-\infty  < x < \infty \} $.  Define
$$ u_{1}(\phi ) = \int _{ \prod _{j} C^{\eta _{j}}_{j}} F_{1}(z)\phi (z) dz,$$
$$u_{2}(\phi ) =\int _{ \prod _{j} C^{-\eta _{j}}_{j}} F_{2}(z)\phi (z) dz$$
for $\phi  \in  {\mathcal T} (T(\mathbb R^{4}))$ and $\eta  =
(\eta _{0}, \ldots , \eta _{3}) \in  \Gamma $.  Let
$$ L = \{ w \in  \mathbb C^{4}; \exists\, x \in  V \  \vert {\rm Re \, }w - x\vert  + \vert {\rm Im \, }w\vert _{1} < \ell \} ,$$
where $V$ is the light-cone and $\vert y\vert _{1} = \vert
y^{0}\vert  + \vert \boldsymbol{y}\vert $.  Assume that $u_{1} -
u_{2} \in  {\mathcal T} (L)^{\prime }$.  Then $F_{i}(z)$ $(i = 1,
2)$ are analytically continued to the set \beq
\label{eq:ell-dist-more} \{ x \in  \mathbb R^{4}; {\rm dist \,
}(x, V) > (\sqrt{2} + 1)\ell \} \eeq and coincide there.\end{cor}
\begin{proof}  The corollary follows from Theorem \ref{lwedge} and the following lemma.
\end{proof}
\begin{lem}  Let $L$ be the set defined in the above
corollary and $g_r(x)$ the function defined by (\ref{eq:gr}) with
$r > \ell /(\sqrt{2} - 1)$.  Then the set $O = \{ x \in \mathbb
R^{4}; g_r(x) > r\} $ contains a set
$$      \{ x \in  \mathbb R^{4}; {\rm dist \, }(x, V) > (\sqrt{2} + 1)\ell \}.$$
The open set
$$      Q = \{ x \in  \mathbb R^{4}; {\rm dist \, }(x, V) > (\sqrt{2} + 3)\ell \} $$
is contained in $O$ and ${\rm dist \, }(\partial O, Q) > 2\ell
$.\end{lem}
\begin{proof}   Observe that
$$      g_r(x) \leq \inf \{ \sqrt{r^{2} + \vert x - {\rm Re \, }w\vert ^{2}} - \vert {\rm Im \, }w\vert ; (w \in  L) \wedge  ({\rm Re \, }w \in  V)\} $$
$$      = \inf \{ \sqrt{r^{2} + \vert x - u\vert ^{2}} - \ell ; u \in  V\}  = \sqrt{r^{2} + y^{2}} - \ell ,$$
where $y = {\rm dist \, }(x, V)$.  Then $g_r(x) > r$ implies
$\sqrt{r^{2} + y^{2}} > r + \ell $ and $y^{2} > \ell (2r + \ell
)$. Since we can choose $r > \ell /(\sqrt{2} - 1)$ arbitrarily
close to $\ell /(\sqrt{2} - 1)$, we have $y^{2} > (\sqrt{2} +
1)^{2}\ell ^{2}$.  This completes the proof.\end{proof}
\begin{rem}
According to Proposition 4.7 of [\cite{BN04}] the $n$-point
function in the difference variables $W_{n-1}(\zeta _{1}, \ldots ,
\zeta _{n-1})$ is analytic in
$$      \{ \zeta  \in  \mathbb C^{4(n-1)}; {\rm Im \, }\zeta _{j} \in  V_{+} + \ell ^{\prime }e,
\  {\rm Im \, }\zeta _{k} \in  V_{+} + Re, \  \ell < \ell ^{\prime
}, k \neq  j\} $$ for sufficiently large $R>0$.  Now for
$$g_{n-2}^{j}(\zeta _{1}, \ldots , \zeta _{j-1}, \hat{\zeta } _{j},
\zeta _{j+1}, \ldots , \zeta _{n-1}) \in {\mathcal T} (T( \mathbb
R^{4(n-2)})),$$ define
$$ F_{1}(\zeta _{j}) = W_{n-1}(\zeta _{j}, g_{n-2}^{j})$$
$$ = \langle W_{n-1}(\zeta _{1}, \ldots , \zeta _{n-1}),
g_{n-2}^{j}(\zeta _{1}, \ldots , \zeta _{j-1}, \hat{\zeta } _{j},
\zeta _{j+1}, \ldots , \zeta _{n-1}) \rangle  ,$$
$$ F_{2}(\zeta _{j}) = W_{n-1}^{j}(\zeta _{j}, g_{n-2}^{j})$$
$$ = \langle W_{n-1}^{j}(\zeta _{1}, \ldots , \zeta _{n-1}),
g_{n-2}^{j}(\zeta _{1}, \ldots , \zeta _{j-1}, \hat{\zeta } _{j},
\zeta _{j+1}, \ldots , \zeta _{n-1}) \rangle  ,$$ where
$$W_{n-1}^{j}(\zeta _{1}, \ldots , \zeta _{n-1}) = W_{n-1}(\zeta
_{1}, \ldots , \zeta _{j-1} + \zeta _{j}, - \zeta _{j}, \zeta _{j}
+ \zeta _{j+1}, \ldots , \zeta _{n-1}).$$  Then $F_{1}$ and
$F_{2}$ satisfy the assumptions of Corollary \ref{application} and
$F_{i}(z)$ $(i = 1, 2)$ are analytically continued to the set
(\ref{eq:ell-dist-more}) and coincide there.  But in [\cite{NB09}]
it is shown that
$$      W_{n-1}(\xi _{j}, g_{n-2}^{j}) - W_{n-1}^{j}(\xi _{j}, g_{n-2}^{j}) = 0$$
in the set \beq  \label{eq:ell-dist} \{ x \in  \mathbb R^{n}; {\rm
dist \, }(x, \bar{V} )
> \ell \}.
\eeq So, we might expect that $F_{i}(z)$ $(i = 1, 2)$ are
analytically continued to each other through the set
(\ref{eq:ell-dist}). We can show that $u_{1} - u_{2}$ is analytic
in the set (\ref{eq:ell-dist}) as follows.  If $x$ belongs to the
set (\ref{eq:ell-dist}), then there exists $\delta > 0$ such that
$\vert x - {\rm Re \, }w\vert > \vert {\rm Im \, }w\vert + \delta
$ for all $w \in L$, and therefore we have $g_{r}(x) > r + \delta
$ for sufficiently small $r > 0$ (see (\ref{eq:gr})).  Then
$U_{12}(z) = (u_{1} - u_{2})*K_{r}(z)$ is holomorphic in the set
(\ref{eq:imaginary}) which contains the set
$$    \{ z = x +
iy \in  \mathbb C^{n}; \vert y\vert < r + \delta , {\rm dist \,
}(x, \bar{V} )
> \ell \} , $$
(see (\ref{eq:imaginary})) and therefore
$$ H_{12}(z) = \int _{\vert \omega \vert =1}d\omega  U_{12}(z + ir\omega )$$
is holomorpic in the set
$$    \{ z = x +
iy \in  \mathbb C^{n}; \vert y\vert < \delta , {\rm dist \, }(x,
\bar{V} )
> \ell \} , $$
which shows that $u_{1} - u_{2}$ is analytic in the set
(\ref{eq:ell-dist}).  But in order to show that $F_{i}(z)$ $(i =
1, 2)$ are analytically continued to each other by our method,
$V_{1} \cup V_{2}$ must contain $B_{r+\delta }$ for some $\delta >
0$, and therefore $\sqrt {2} r - \ell > r$ must be hold (see
figure 2), i.e., $r > \ell /(\sqrt{2} - 1)$.  This is the reason
why our method can only show that $F_{i}(z)$ $(i = 1, 2)$ are
analytically continued to each other through the set
(\ref{eq:ell-dist-more}).
\end{rem}

\bibliographystyle{plain}
\bibliography{hfqft20080628}

\begin{thebibliography}{10}

\bibitem{BV60}
N.~N. Bogoliubov and V.~S. Vladimirov.
\newblock On some mathematical problem of quantized field theory.
\newblock In {\em Proceedings of the International Congress of Mathematicians,
  Edinburgh 1958}, New York, 1960. Cambridge University Press.

\bibitem{BN89}
E.~Br{\"u}ning and S.~Nagamachi.
\newblock Hyperfunction quantum field theory: Basic structural results.
\newblock {\em J. Math. Physics}, 30:2340--2359, 1989.

\bibitem{BN04}
E.~Br{\"u}ning and S.~Nagamachi.
\newblock Relativistic quantum field theory with a fundamental length.
\newblock {\em J. Math. Phys.}, 45:2199--2231, 2004.

\bibitem{BN08}
E.~Br\"uning and S.~Nagamachi.
\newblock Solutions of a linearized model of {H}eisenberg's fundamental
  equation {II}.
\newblock {\em J. Math. Phys.}, 49:052304--1 -- 052304--22, 2008.

\bibitem{Ep60}
H.~Epstein.
\newblock Generalization of the ``edge of the wedge'' theorem.
\newblock {\em J. Math. Phys.}, 1:524--531, 1960.

\bibitem{Fr06}
D.~H.T. Franco.
\newblock The edge of the wedge theorem for tempered ultrahyperfuncions.
\newblock {\em arXiv:math/0609751}, 2006.

\bibitem{Fr07}
D.~H.T. Franco.
\newblock The edge of the wedge theorem for tempered ultrahyperfuncions {II}, a
  generalized version.
\newblock {\em arXiv:0708.0252}, 2007.

\bibitem{Ha61}
M.~Hasumi.
\newblock Note on the $n$-dimensional tempered ultra-distributions.
\newblock {\em Tohoku Math. J.}, 13:94--104, 1961.

\bibitem{Hoe83}
L.~H{\"o}rmander.
\newblock {\em The analysis of linear partial differential operators {I}},
  volume 256 of {\em Grundlehren der mathematischen Wissenschaften}.
\newblock Springer-Verlag, Berlin Heidelberg New York Tokyo, 1983.

\bibitem{Ka88}
A.~Kaneko.
\newblock {\em Introduction to Hyperfunctions}.
\newblock Mathematics and Its Applicatons (Japanese Series). Kluwer Academic
  Publishers, Dordrecht Boston London, 1988.

\bibitem{Ka70}
T.~Kawai.
\newblock On the theory of {F}ourier hyperfunctions and its applications to
  partial differential equations with constant coefficients.
\newblock {\em J. Fac. Sci. Univ. Tokyo, Sect. I.A}, 17:467--517, 1970.

\bibitem{Ma64}
A.~Martineau.
\newblock Distributions et valeurs au bord des fonctions holomorphes.
\newblock In {\em Proc. Intern. Summer Course on the Theory of Distributions},
  pages 195--326, Lisbon, 1964.

\bibitem{Ma67}
A.~Martineau.
\newblock {\em Th\'eor\`eme sur le prolongement analytique du type ``Edge of
  the Wedge Theorem''}.
\newblock S\'eminair Bourbaki, 20-i\`eme ann\'ee, No. 340, 1967/68.

\bibitem{Ma70}
A.~Martineau.
\newblock Le ``edge of the wedge theorem'' en th\'eorie des hyperfonctions de
  sato.
\newblock In {\em Proc. Intern. Conf. on Functional Analysis, Tokyo, 1969},
  pages 95--106, Tokyo, 1970. Univ. Tokyo Press.

\bibitem{Mo69}
M.~Morimoto.
\newblock Sur les ultradistributions cohomologiques.
\newblock {\em Ann. Inst. Fourier}, 19:129--153, 1969.

\bibitem{Mo72}
M.~Morimoto.
\newblock La d\'ecomposition de singularit\'es d'ultradistributions
  cohomologiques.
\newblock {\em Proc. Japan Acad.}, 48:129--153, 1972.

\bibitem{Mo73}
M.~Morimoto.
\newblock Edge of the wedge theorem and hyperfunction.
\newblock In {\em Hyperfunctions and pseudo-differential equations (Proc.
  Conf., Katata, 1971}, pages 41--81, Berlin, 1973. Lecture Notes in Math.,
  Vol. 287, Springer.

\bibitem{Mo75b}
M.~Morimoto.
\newblock Convolutors for ultrahyperfunctions.
\newblock In {\em International Symposium on Mathematical Problems in
  Theoretical Physics}, volume~39 of {\em Lecture Notes in Phys.}, pages
  49--54, Berlin, 1975. Springer.

\bibitem{Mo93}
M.~Morimoto.
\newblock {\em An introduction to Sato's hyperfunctions}.
\newblock American Mathematical Society, 1993.

\bibitem{NB08a}
S.~Nagamachi and E.~Br\"uning.
\newblock Solutions of a linearized model of {H}eisenberg's fundamental
  equation {I}.
\newblock {\em arXiv:0804.1663 [math-ph]}, 2008.

\bibitem{NB09}
S.~Nagamachi and E.~Br\"uning.
\newblock Frame {I}ndependence of the {F}undamental {L}ength in {R}elativistic
  {Q}uantum {F}ield {T}heory.
\newblock Preprint, 2009.

\bibitem{NN93}
S.~Nagamachi and T.~Nishimura.
\newblock Edge of the wedge theorem for {F}ourier hyperfunctions.
\newblock {\em Funkcialaj Ekvacioj}, 36:499--517, 1993.

\bibitem{NN01}
T.~Nishimura and S.~Ngamachi.
\newblock Support and kernel theorem for {F}ourier hyperfunctions.
\newblock {\em Osaka J. Math.}, 38:667--680, 2001.

\bibitem{Su04}
M.~Suwa.
\newblock Distributions of {E}xponential {G}rowth with {S}upport in a {P}roper
  {C}one.
\newblock {\em Publ. RIMS, Kyoto Univ.}, 40:565 -- 603, 2004.

\bibitem{Vl64}
V.~S. Vladimirov.
\newblock {\em Methods of the {T}heory of {F}unctions of {M}any {C}omplex
  {V}ariables}.
\newblock The M.I.T Press, Cambridge, Massachussetts, London, 1964.

\end{thebibliography}

\end{document}